\documentclass[aps,prl,twocolumn,a4paper,amsfonts,amssymb,amsmath,floatfix,nofootinbib,superscriptaddress]{revtex4-2}
\pdfoutput=1
\usepackage[utf8]{inputenc}
\usepackage[english]{babel}
\usepackage[colorlinks=true,allcolors=blue]{hyperref}
\usepackage{graphicx}
\usepackage{soul}
\usepackage{amsthm}
\usepackage{multirow}
\usepackage{mathrsfs}
\usepackage{bm}
\usepackage{dsfont}
\usepackage{xparse}
\usepackage{xcolor}
\usepackage{textcase}
\usepackage{url}
\usepackage{lipsum}
\usepackage{epstopdf}
\usepackage{footnote}
\usepackage{color}  
\newtheorem*{lemma*}{Lemma}
\newtheorem*{def*}{Definition}
\newtheorem{thm}{Theorem}

\newtheorem{lemma}[thm]{Lemma}

\usepackage{enumitem}
\usepackage{float}

\def\g{\gamma}

\def\e{\epsilon}
\def\j/{\varphi}
\def\l{\lambda}

\def\m{\mu}
\def\s{\sigma}

\def\j{\varphi}

\def\q{\theta}
\def\r{\rho}
\def\n{\nu}
\def\ox{\otimes}
\def\>{\rangle}

\begin{document}
\title{Hilbert space separability and the Einstein-Podolsky-Rosen state}
\author{Miguel Gallego}
\email{miguel.gallego.ballester@univie.ac.at}
\affiliation{University of Vienna, Faculty of Physics, Vienna Center for Quantum Science and Technology, Boltzmanngasse 5, 1090 Vienna, Austria}
\date{\today}

\begin{abstract}
Quantum mechanics is formulated on a Hilbert space that is assumed to be separable. However, there seems to be no clear reason justifying this assumption. Does it have physical implications? We answer in the positive by proposing a test that witnesses the non-separability of the Hilbert space, at the expense of requiring measurements with uncountably many outcomes. In the search for a less elusive manifestation of non-separability, we consider the original Einstein-Podolsky-Rosen (EPR) state as a candidate for possessing nonlocal correlations stronger than any state in a separable Hilbert space. Nevertheless, we show that, under mild assumptions, this state is not a vector in any bipartite space, even non-separable, and therefore cannot be described within the standard Hilbert space formalism. 
\end{abstract}

\maketitle

Since von Neumann \cite{vonneumann}, quantum theory is formulated in terms of a complex Hilbert space that is assumed to be \emph{separable}, meaning that it contains a countable dense subset or that, equivalently, it admits a countable orthonormal basis. Essentially, the dimension of the Hilbert space is taken to be either finite, in which case it is isomorphic to $\mathbb{C}^n$ for some $n\in \mathbb{N}$, or countably infinite, in which case it is isomorphic to the space of square-summable sequences $\ell^2(\mathbb{N})$, in turn isomorphic to the space of (equivalence classes of) square-integrable functions $L^2(\mathbb{R})$. But no clear physical reason seems to motivate the assumption of separability, and even though the attempt to treat quantum systems with non-separable Hilbert spaces may raise problems that are not only technical \cite{earman}, nothing seems to fully justify their dismissal.

In a paper of Halvorson \cite{halvorsoncomplementarity}, the author discusses two representations of the canonical commutation relations on $\ell^2(\mathbb{R})$, the non-separable Hilbert space of square-summable functions from $\mathbb{R}$ into $\mathbb{C}$: one where the position operator has eigenstates, at the expense that the momentum operator does not exist, and one where the momentum operator has eigenstates, at the expense that the position operator does not exist. This is unlike the Schrödinger representation in $L^2(\mathbb{R})$, where both operators exist and there are states with arbitrarily small spread in either position or momentum, but never zero. Inspired by Halvorson's work, we present two main results. In the first part we propose a simple \emph{gedankenexperiment} for which separable and non-separable Hilbert spaces make different physical predictions. The experiment has the form of a prepare-and-measure scenario which works as a dimension witness, distinguishing between countable and uncountable dimension. Its downside is that the measurements involved are required to resolve an uncountable number of outcomes, which is usually thought to be impossible. One could then look for other tests of non-separability, like for instance in Bell-type experiments \cite{bell}. Namely, one could ask: are quantum correlations in non-separable Hilbert spaces stronger than those in separable spaces or, in other words, do they violate  Tsirelson's bound \cite{tsirelson80, tsirelson87}? As far as we know, this is an open question. A potential candidate \cite{halvorsonepr} for such ``stronger-than-separable" correlations is the original EPR state \cite{epr}, a simultaneous eigenstate of the sum of two parties' positions and the difference of their momenta. Being an eigenstate of continuous quantities, however, it cannot be realized as a vector in a separable Hilbert space. But given the results of Halvorson \cite{halvorsoncomplementarity}, one may think that it can be realized as a vector in a non-separable Hilbert space, just like a simultaneous eigenstate of, for example, the first party's position and the second party's momentum can be constructed in a tensor product of two of Halvorson's representations. In the second part of the paper we show that the EPR state is not a vector in any tensor product representation where one of the factors is a Halvorson-type representation (one where some phase-space observable has eigenvectors), and hence cannot be represented in the standard Hilbert space formalism. The conjectured stronger-than-separable correlations must therefore be looked for elsewhere.

\section{A test of Hilbert space separability}

Consider the following prepare-and-measure scenario (see Figure \ref{fig}). Alice receives an input $x$ drawn at random from a set $X$. She then prepares a quantum system in a state $\psi_x \in \mathcal{H}$, where $\mathcal{H}$ is a Hilbert space, and sends it to Bob, who measures the system obtaining an output $y$ in the same set $X$ as Alice. Let $G$ denote the average probability that Bob guesses Alice's input, i.e. that $y=x$. Then, in the case where the set $X$ has finite size $|X| \in \mathbb{N}$, the quantity $G$ can be used as a dimension witness as follows.  If the Hilbert space $\mathcal{H}$ has dimension $n \geq |X|$, Alice can encode her input in one of a number $|X|$ of orthogonal states, which Bob can distinguish perfectly. This strategy thus gives $G=1$. If, however, the dimension is $n < |X|$, then Bob cannot distinguish with unit probability between $|X|$ different states in $\mathcal{H}$, and $G<1$. In particular, $G \leq n/|X|$, as shown in \cite{brunnernavascues}. Our idea is to generalize this result in order to witness an uncountable dimension.

\begin{figure}
        \centering        \includegraphics[width=0.3\textwidth]{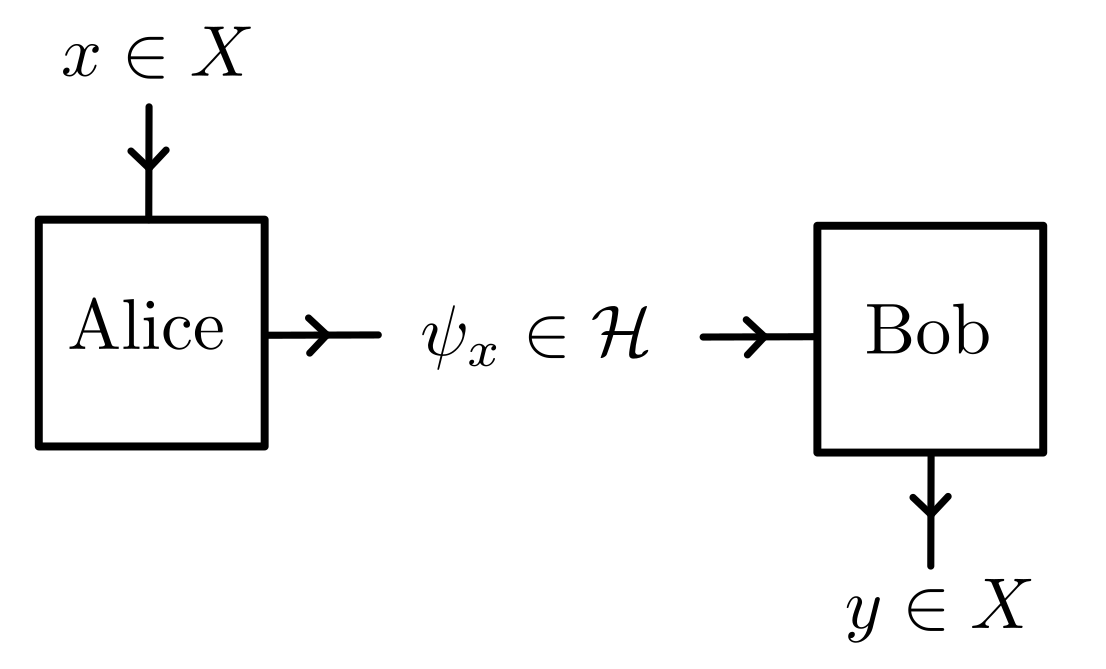}
        \caption{\textbf{Prepare-and-measure scenario.} Alice receives an input $x\in X$, prepares a quantum system in a state $\psi_x \in \mathcal{H}$ and sends it to Bob, who performs a measurement on the system, obtaining an outcome $y \in X$ in the same set as Alice.}
        \label{fig}
\end{figure}

For this purpose, let $X$ be \emph{uncountable} and let us turn it into a measurable space $(X, \Sigma)$ by equipping it with a $\sigma$-algebra $\Sigma$ of subsets. Moreover, let us model Bob's measurement on the state $\psi_x \in \mathcal{H}$ by means of a \emph{positive operator-valued measure} 
(POVM) on the same measurable space $(X,\Sigma)$ as Alice, i.e. a function $E: \Sigma \to \mathfrak{B}(\mathcal{H})$ taking measurable sets to bounded linear operators on $\mathcal{H}$ such that the map $S \mapsto \langle  \psi ,  E(S) \psi\rangle$ defines a probability measure on $(X,\Sigma)$ for every unit vector $\psi \in \mathcal{H}$. Then, the probability that Bob's outcome $y$ lies in $S \in \Sigma$ given that he receives the state $\psi_x$ is
\begin{align*}
    \text{Pr}(y \in S \, | \, \psi_x ) = \langle \psi_x, E(S) \psi_x  \rangle \, .
\end{align*}
Now, assuming that all singletons are measurable, i.e. that $\{ x \} \in \Sigma$ for all $x \in X$, we can define the probability that Bob guesses Alice's input $x$ as $g(x) := \text{Pr}(y = x | \psi_x )$ or, in terms of the POVM elements,
\begin{align*}
    g(x) = \langle \psi_x , E(\{x \} ) \psi_x  \rangle \, .
\end{align*}
The average guessing probability would be the integral of the function $g$ against some probability measure on $(X,\Sigma)$, but we state our result directly in terms of $g$:
\begin{thm}
    If the Hilbert space $\mathcal{H}$ is separable, then $g=0$ except possibly on a countable subset, while $g \equiv 1$ (identically) can be attained in a non-separable Hilbert space.
\label{thm:test}
\end{thm}
\begin{proof}
First, let $\mathcal{H} = \ell^2(X)$ be the space of functions $\psi: X \to \mathbb{C}$ such that $\psi(x) \neq 0$ for countably many $x\in X$ and $\sum_{x \in X} |\psi(x)|^2 < \infty$. This forms a non-separable Hilbert space with the inner product $\langle \psi, \phi \rangle := \sum_{x\in X}  \psi^*(x) \phi(x)$. Upon receiving the input $x\in X$, let Alice prepare the state $\psi_x \in \ell^2 (X)$ given by the characteristic function of $\{ x \}$, i.e
\begin{align*}
    \psi_x (x') = \begin{cases} 1 & \text{ if } x' = x  \, ,\\
    0 & \text{ otherwise} \, .
    \end{cases}
\end{align*}
Let Bob implement a POVM where $E(S)$ is the projection onto the closed linear span of $\{ \psi_x \}_{x \in S}$.\footnote{Indeed, this defines a probability measure: since $\{ \psi_x \}_{x \in X}$ is an orthonormal basis for $\ell^2 (X)$, the projection $E(S)$ for $S$ a disjoint union of sets $S_i$ is the sum of projections $E(S_i)$, so it defines a measure. Moreover, $E(X)$ is the identity in $\ell^2 (X)$, so it is normalized.} Then, for each input $x\in X$, the guessing probability is
\begin{align*}
    g(x)  = \langle \psi_x ,  E(\{ x \} ) \psi_x  \rangle \equiv 1 \, .
\end{align*}

On the other hand, let $\mathcal{H}$ be separable. We will show that the set $\{ g \neq 0 \}: = \{ x \in X : g(x) \neq 0 \}$ is countable. Let $ \{ E \neq 0 \}:= \{ x \in X : E(\{x\}) \neq 0 \}$, which clearly contains $ \{ g \neq 0\}$. By positivity of $E(\{x \})$, for every $x \in \{ E \neq 0 \}$ there exists some $\phi_x \in \mathcal{H}$ such that $\langle \phi_x , E(\{x \}) \phi_x \rangle > 0$. By continuity of $E(\{x \})$, for every such $\phi_x$ there exists an open set $U_x \subset \mathcal{H}$ containing $\phi_x$ such that $\langle \phi  ,  E(\{x \}) \phi \rangle > 0$ for all $\phi \in U_x$. By separability of $\mathcal{H}$, there exists a countable dense subset $\{ \chi_n \}_{n \in \mathbb{N}} \subseteq \mathcal{H}$, which we assume (without loss of generality) to have unit norm. Being dense, this subset necessarily has non-empty intersection with each $U_x$, so that to each $x \in \{ E \neq 0 \}$ we can assign $n \in \mathbb{N}$ such that $\chi_n \in U_x$. Now suppose the set $\{ E\neq 0 \}$ is uncountable. Then, uncountably many values of $x$ have to be assigned to the same $n$, so that, for uncountably many $x$, we have $\chi_n \in U_x$, and in particular $\langle \chi_n  , E(\{ x\} ) \chi_n \rangle > 0$. But this means that the map $S \mapsto \langle \chi_n , E(S) \chi_n \rangle$ cannot define a normalized measure, since the sum of $\langle \chi_n ,  E( \{ x \} ) \chi_n \rangle$ over $\{E\neq 0\}$ diverges. In conclusion, $\{ E \neq 0 \}$ must be countable, and so must be $\{ g \neq 0 \}$.\footnote{This result is not altered if Alice encodes her message into mixed states, described by density matrices $\rho_x$, since in this case the guessing probability is given by $g(x) = \text{tr} \, E(\{ x \} ) \rho_x$ and  anyway $\{ g \neq 0\} \subseteq \{ E \neq 0 \}$ holds.}
\end{proof}

Several comments are in order. First, notice that, if a system has a separable Hilbert space $\mathcal{H}$, then the space $\mathcal{H}^{\otimes n}$ associated to a finite number $n$ of them is also separable, and so is the Fock space $\oplus_{n\in \mathbb{N}} \mathcal{H}^{\otimes n}$. However, as long as $\dim{\mathcal{H}}\geq 2$, the Hilbert space $\otimes_{n \in \mathbb{N}} \mathcal{H}$ associated to an infinite number of particles (also called \emph{infinite spin chain}) is non-separable with dimension $2^{\aleph_0}$ \cite{earman}. Therefore, the above prepare-and-measure scenario can be reformulated as a witness of the non-finiteness of the spin chain. Namely, let the set $X$ have cardinality $2^{\aleph_0}$, and let the state $\psi_x$ sent to Bob be the state of a spin chain, where each spin has a separable Hilbert space $\mathcal{H}$. We then know that $g=0$ (except possibly on a countable subset) if the chain is finite, while $g\equiv 1$ can be attained if it is infinite by translating the strategy given above from $\ell^2(X)$ to $\otimes_{n \in \mathbb{N}} \mathcal{H}$, which essentially ammounts to choosing a bijection between $X$ and an orthonormal basis of $\otimes_{n \in \mathbb{N}} \mathcal{H}$.

Second, the content of Theorem \ref{thm:test} is not specifically about quantum mechanics, but rather about the amount of available distinguishable states. Is it possible to formulate and prove a similar but more general statement, applying both to classical and quantum mechanics? Another paper by Halvorson \cite{halvorsoncontinuous} might provide a key for this. As he argues, the logic of experimental propositions for one continuous quantity (like position or momentum) is equal both in classical and quantum mechanics, and particles possessing a sharp value of this quantity can be described in both cases by singular (i.e. finitely additive but not $\s$-additive) states. Thus, reformulating the prepare-and-measure scenario described before in terms of the logic of experimental propositions for some continuous physical quantity, a possible generalization of our result would be something like:
\begin{quote}
    `If the states are $\s$-additive, then $g=0$ (except possibly on a countable subset), while $g\equiv 1$ (identically) can be attained with singular states.'
\end{quote}
We do not attempt to prove this here.

Finally, as we have announced, the proposed experiment looks rather idealized, since Alice is required to have access to an uncountable number of preparations and Bob is required to perfectly distinguish between them or, reformulated in terms of the infinite spin chain, they need to be able to prepare and measure a (countable) infinity of particles. It may then seem that the result is not robust, in the sense that the gap will disappear if we relax the conditions of the game and allow Bob to guess Alice's input to within some finite accuracy, i.e. counting as correct all of Bob's guesses at a distance $\epsilon > 0$ from Alice's input. But, strictly speaking, this depends on which metric is used to measure the distance $\e$. Let us therefore equip the set $X$ with a metric $d: X \times X \to \mathbb {R}$ and define the guessing probability as $g_\epsilon(x) := \text{Pr} ( y \in B_\epsilon (x) | \psi_x )$, where $B_\epsilon(x) = \{ x' \in X : \, d(x,x')< \epsilon \}$ is the open ball of radius $\epsilon$ centered at $x$. Then, if the metric space $(X,d)$ is \emph{separable} (as is for instance the set of real numbers with the usual metric), $g_\epsilon \equiv 1$ can be attained on a separable Hilbert space, since countably many orthogonal states suffice to encode a dense subset of $(X,d)$, so that the gap disappears. If, on the other hand, the space $(X,d)$ is \emph{not} separable (however esoteric that may seem\footnote{Two examples of such non-separable metric spaces are $\mathbb{R}$ with the discrete metric (assigning $0$ to the distance of any point to itself and $1$ to the distance between any two distinct points) or the space of bounded sequences of complex numbers $\ell^\infty(\mathbb{N})$ with the usual metric induced by the norm $\| x \|_\infty = \sup_{n \in \mathbb{N}} | x_n |$.}), it can be shown that any strategy on a separable Hilbert space gives $g_\epsilon =0$ (except possibly on a countable subset), making the witness robust. But again, there is nothing inherently quantum here: the same argument can be used in favor of a potential experiment that could determine that the value of a certain continuous classical quantity was some exact real number. This is generally accepted to be impossible. As Birkhoff and von Neumann \cite{birkhoffvonneumann} put it: “\emph{how absurd it would be to call an `experimental proposition,' the assertion that the angular momentum of the earth around the sun was at a particular instant a rational number!}”. To us, however, this statement (see also \cite{vonneumannoperatorenmethode}) is not obvious, and we believe a more careful analysis of the assumptions involved would be necessary to gain a deeper understanding.

\section{Is the EPR state a vector in some Hilbert space?}
As a possible alternative to witness the non-separability of the Hilbert space, we consider the potential stronger-than-separable quantum correlations of the EPR state. This state can be defined as the algebraic state that assigns dispersion-free values to the sum of two parties' position operators and the difference of their momenta, and it can be shown to be unique and pure \cite{halvorsonepr}. The EPR state can be straightforwardly represented as a vector in the (non-separable) Hilbert space $\ell^2(\mathbb{R}) \ox \ell^2(\mathbb{R})$ using the Gelfand-Naimark-Segal construction \cite{kadisonringrose}. But this is unsatisfactory, since in this case one Hilbert space corresponds to the degree of freedom of the sum of the positions and the other corresponds to the degree of freedom of the difference of the momenta, and therefore the tensor product between the two Hilbert spaces does not correspond to the system bipartition \cite{huang}. We therefore regard this case as ``trivial'', and we ask (more precisely) whether the EPR state a is vector in some \emph{non-trivial} or \emph{tensor product representation}, i.e. one such that Alice's operators act trivially on Bob's Hilbert space and and vice versa. As far as we know, this is an open question. Here we show that in any such tensor product representation, if one of the factors is a Halvorson-type representation (one where some phase-space operator has eigenvectors) then the EPR does not exist as a vector. This suggests that the EPR state cannot be described within the standard Hilbert space formalism, and therefore it cannot be used to witness the non-separability of the Hilbert space. 

In order to define the EPR state in Hilbert space, let $W(a,0)$ and $W(0,b)$ be one-parameter (not necessarily continuous) families of unitary bounded operators satisfying the Weyl \emph{canonical commutation relations} (CCRs):
\begin{align*}
    W(a_1, a_2) W(b_1, b_2) = e^{i(a_1b_2-a_2 b_1)/2} \, W(a_1+b_1, a_2+b_2) \, ,
\end{align*}
for all $a_1, a_2, b_1, b_2 \in \mathbb{R}$, so that $W(1,0)$ corresponds to the exponential of the position operator and $W(0,1)$ to the exponential of the momentum operator (in case these exist). Then, if $\mathcal{H}_A$ and $\mathcal{H}_B$ are Hilbert spaces corresponding to two parties, we say a vector $\psi \in \mathcal{H}_A \ox \mathcal{H}_B$ is an EPR state (in a tensor product representation of the CCRs) if
\begin{enumerate}[label=(\textit{\roman*})]
    \item $W_A(a,0) \otimes W_B(a,0) \psi = e^{i a x} \psi \, , \quad \forall a \in \mathbb{R}$ ,
    \label{condii}
    \item $W_A(0,b) \otimes W_B(0,-b) \psi = e^{i b p} \psi \, , \quad \forall b \in \mathbb{R}$ ,
    \label{condiii}
\end{enumerate}
for some $x,p \in \mathbb{R}$ (corresponding, respectively, to the value of the sum of positions and the difference of momenta) \cite{halvorsonepr, huang}.

\begin{thm}
    If, for some $\q \in [0, 2 \pi]$, the operator $W_A(\cos \q, \sin \q)$ has eigenvectors in $\mathcal{H}_A$ and these form an orthonormal basis, then there is no EPR state in $\mathcal{H}_A \ox \mathcal{H}_B$. 
\label{thm:epr}
\end{thm}

Before providing the proof of this theorem, let us prove the following lemma (see also the second theorem in \cite{halvorsoncomplementarity}):
\begin{lemma}
    If $W(a,0)$ has eigenvectors and they form an orthonormal basis of $\mathcal{H}$, then there are no eigenvectors of $W(0,b)$.
\end{lemma}
\begin{proof}
    Let $ \{ \chi_\lambda : \lambda \in \mathbb{R} \}$ be the orthonormal basis for $\mathcal{H}$ given by position eigenvectors, i.e. $W(a,0) \chi_\lambda = f(a,\lambda)  \chi_\lambda$, for all $a \in \mathbb{R}$ and $\lambda \in \mathbb{R}$. By unitarity, $f$ must be a phase. By the group property $W(a,0) W(b,0) = W(a+b,0)$, the phase must be linear in $a$, so that $W(a,0) \chi_\lambda = e^{i a \varphi(\lambda)}  \chi_\lambda$. On the other hand, from the Weyl CCRs, one can check that the action of $W(0,b)$ on the same basis is given by
    \begin{align*}
        \langle \chi_\m ,  W(0,b) \chi_\l \rangle = e^{iab } e^{i a \j (\l) } e^{- i a \j ( \m)} \langle  \chi_\m ,  W(0,b) \chi_\l \rangle \, .
    \end{align*}
    Then, if $\langle \chi_\m  , W(0,b) \chi_\l \rangle \neq 0$, we have that
    \begin{align*}
        1 = e^{i ab} e^{i a \j (\l) } e^{- i a \j ( \m)} \, , \quad \forall a \in \mathbb{R} \, .
    \end{align*}
    Therefore, we have shown that
    \begin{align}
        \langle  \chi_\m , W(0,b) \chi_\l \rangle \neq 0 \, \, \Rightarrow \, \, b = \j(\m) - \j(\l)  \, \, \text{mod} \, \, 2 \pi \, .
    \label{cond}
    \end{align}
    Now suppose there exists an eigenvector of momentum $\psi \in \mathcal{H}$. Since it is a vector, it can be written in the position basis as $\psi = \sum_{\l \in \mathbb{R}} \langle  \chi_\l , \psi \rangle \chi_\l$, where $K := \{ \l \in \mathbb{R} : \langle  \chi_\l , \psi \rangle \neq 0 \}$ is countable. On the other hand, being a momentum eigenvector it must satisfy $W(0,b) \psi = e^{i b \g} \psi$ for all $b \in \mathbb{R}$ and for some $\g \in \mathbb{R}$, which can be written as
    \begin{align*}
        \sum_{\l \in \mathbb{R}} \langle  \chi_\l , \psi \rangle \, W(0,b) \chi_\l = e^{i b \g} \sum_{\m \in \mathbb{R}} \langle  \chi_\m , \psi \rangle \, \chi_\m \, , \quad \forall b \in \mathbb{R} \, .
    \end{align*}
    Now, let us pick some $\m \in K$ and take the inner product of the previous equation with $\chi_\m$, obtaining
    \begin{align*}
        \sum_{\l \in \mathbb{R}} \langle \chi_\l , \psi \rangle \, \langle \chi_\m  , W(0,b) \chi_\l \rangle = e^{i b \g}  \langle \chi_\m , \psi \rangle \neq 0 \, , \quad \forall b \in \mathbb{R} \, .
    \end{align*}
    In particular there is always a non-zero element in the sum of the left-hand side, so that, for every $\m \in K$, there is some $\l \in K$ such that 
    \begin{align*}
        \langle \chi_\m , W(0,b) \chi_\l  \rangle \neq 0 \, .
    \end{align*}
    In turn, by \eqref{cond}, this implies $b = \j(\m)- \j(\l) \, \, \text{mod} \, \, 2 \pi$. So we have shown that the existence of an eigenvector of momentum implies that $\forall b \in \mathbb{R}$ and $\forall \m \in K$, there exists $\l \in K$ such that
    \begin{align*}
        b = \j (\m) - \j(\l) \, \,  \text{mod} \, \, 2 \pi \, .
    \end{align*}  
    But this is impossible: $b$ on the left hand side runs through an uncountable set ($\mathbb{R}$), whereas $\j$ on the right hand side is only evaluated on a countable set ($K$), and therefore there cannot be a one-to-one correspondence between elements of one type and elements of the other.
\end{proof}
This lemma can be generalized to arbitrary non-parallel directions in phase space: if $W(\cos \q, \sin \q)$ (which corresponds to the exponential of the phase-space observable at an angle $\q$) has eigenvectors and they form an orthonormal basis, then for any different angle $\q'$ the operator $W(\cos \q', \sin \q')$ has no eigenvectors. With this we are ready to provide the proof of the theorem. \\

\begin{proof}[Proof of Theorem \ref{thm:epr}]
    Let us assume, without loss of generality, that $x=p=0$, and let us falsify a necessary condition for \ref{condii} and \ref{condiii}, namely
\begin{align}
    \Big( W_A(a,0) \otimes W_B(a,0) \Big) \Big( W_A(0,b) \otimes W_B(0,-b) \Big) \psi = \psi 
\label{star}
\end{align}
for all  $a,b \in \mathbb{R}$. Let us introduce polar coordinates $a = r \cos \q$ and $b = r \sin \q$, and let $\{ \chi_\l : \l \in \mathbb{R} \}$ be an orthonormal basis for $\mathcal{H}_A$ of $\q$-eigenvectors, i.e.
\begin{align*}
    W_A(r \cos \q ,r \sin \q ) \chi_\l = e^{i r \j ( \l )} \chi_\l \, , \quad \forall r \in \mathbb{R} \, .
\end{align*}
Again, the Weyl CCRs imply that in the orthogonal direction we have the implication
\begin{align}
    & \langle \chi_\m , W_A(- r \sin \q ,r \cos \q ) \chi_\l \rangle \neq 0 \nonumber \\
    & \qquad \qquad \qquad  \Rightarrow \, \,  r = \j(\m) - \j(\l) \, \, \text{mod} \, \, 2 \pi  \, .
\label{condtheta}
\end{align}
Let $\{ \xi_\l : \l \in \mathbb{R} \}$ be some arbitrary orthonormal basis for $\mathcal{H}_B$. Then $\{ \chi_\l \otimes \xi_\m : \l, \m \in \mathbb{R} \}$ is an orthonormal basis for $\mathcal{H}_A \otimes \mathcal{H}_B$. Writing $\psi$ in this basis we have
\begin{align*}
    \psi = \sum_{\l,\m \in \mathbb{R}} \langle  \chi_\l \otimes \xi_\m , \psi \rangle \, \chi_\l \otimes \xi_\m \, ,
\end{align*}
and \eqref{star} reads
\begin{widetext}
\begin{align}
    \sum_{\l,\m \in \mathbb{R}}   \langle  \chi_\l \otimes \xi_\m  , \psi \rangle \,  W_A(a,0) W_A(0,b) \chi_\l \otimes W_B(a,0) W_B(0,-b) \xi_\m =  \sum_{\nu,\r \in \mathbb{R}}  \langle  \chi_\n \otimes \xi_\r , \psi \rangle \, \chi_\n \otimes \xi_\r \, , \quad \forall a ,b \in \mathbb{R} \, .
\label{starp}
\end{align}
\end{widetext}
Let $K:= \{ (\l,\m) \in \mathbb{R}^2 : \langle  \chi_\l \otimes \xi_\m  , \psi \rangle \neq 0 \}$, which is countable, and let us take the inner product of \eqref{starp} with $\chi_\n \otimes \xi_\r$, where $(\n,\r) \in K$, obtaining
\begin{widetext}
\begin{align*}
    \sum_{\l,\m \in \mathbb{R}}   \langle \chi_\l \otimes \xi_\m  , \psi \rangle \,  \langle \chi_\n ,  W_A(a,b) \chi_\l \rangle \,  \langle \xi_\r , W_B(a,-b) \xi_\m  \rangle  =   \langle  \chi_\n \otimes \xi_\r , \psi \rangle \neq 0 \, , \quad \forall a,b \in \mathbb{R} \, .
\end{align*}
\end{widetext}
Then, in particular, for every $a,b \in \mathbb{R}$ there is a nonzero element in the sum, i.e. for every $(\n,\r) \in K$ there exists $(\l,\m) \in K$ such that  
\begin{align*}
    \langle \chi_\n , W_A(a,b) \chi_\l  \rangle \neq 0 \, .
\end{align*}
Defining $K_A:= \{ \l \in \mathbb{R} : \exists \m \in \mathbb{R} : \langle \chi_\l \otimes \xi_\m ,  \psi \rangle \neq 0 \}$, which is also countable, the existence of the EPR state implies that for all $a,b \in \mathbb{R}$ and for all $\n \in K_A$ there exists $\l \in K_A$ such that 
\begin{align*}
     \langle  \chi_\n, W_A(a,b) \chi_\l \rangle \neq 0 \, .
     \label{eq}
\end{align*}
Again, by \eqref{condtheta}, and since $K_A$ is countable, this is impossible.
\end{proof}

In conclusion, the EPR state cannot be described within the standard Hilbert space formalism, at least for the family of representations of the CCRs considered. It remains an open question whether a different representation of the CCRs that respects the system bipartition contains an EPR state, and whether such a state (or any other in Hilbert space) possesses stronger-than-separable correlations.

\emph{Acknowledgments.}---  We would like to thank Sebastian Horvat and Borivoje Daki\'c for their helpful comments and ideas. The author acknowledges support from  the Austrian Science Fund (FWF) through BeyondC-F7112 and from the ESQ Discovery programme (Erwin Schrödinger Center for Quantum Science \& Technology), hosted by the Austrian Academy of Sciences
(ÖAW).
\bibliography{references}

\end{document}